\documentclass[12pt]{article}
\usepackage[utf8]{inputenc}
\usepackage{fullpage}
\usepackage{amsmath}
\usepackage{amsthm}
\usepackage{xcolor}
\usepackage{authblk}
\usepackage{hyperref}
\usepackage{amssymb}
\usepackage{cleveref} 

\theoremstyle{definition}
\newtheorem{theorem}{Theorem}
\newtheorem{definition}{Definition}
\newtheorem{lemma}{Lemma}[section]
\newtheorem{corollary}[lemma]{Corollary}

\newcommand{\QCPH}{\mathsf{QCPH}}
\newcommand{\BQP}{\mathbf{BQP}}

\newcommand{\pQCMA}{\text{\sc Precise}\mathbf{QCMA}}
\newcommand{\QCMA}{\mathsf{QCMA}}

\newcommand{\QPH}{\mathsf{QPH}}
\newcommand{\Qs}[1]{\mathsf{QC}\Sigma_{#1}}
\newcommand{\Qf}[1]{\mathsf{QC}\Pi_{#1}}
\newcommand{\abrak}[1]{\left\langle #1 \right \rangle}
\renewcommand{\P}{\mathbf{P}}

\newcommand{\cc}[1]{\mathsf{#1}}
\newcommand{\class}[1]{\mathcal{#1}}

\newcommand{\secref}[1]{\S \ref{#1}}

\newcommand{\defeq}{\overset{\mathrm{def}}{=\joinrel=}}

\begin{document}

\title{A Collapsible Polynomial Hierarchy for Promise Problems}
\author[]{Chirag Falor$^\S$\thanks{cfalor@mit.edu}}
\author[]{Shu Ge$^\S$\thanks{geshu@mit.edu} }

\author[]{Anand Natarajan\thanks{anandn@mit.edu}}
\affil[]{Massachusetts Institute of Technology}
\date{November 2023}

\maketitle
\def\thefootnote{$\S$}\footnotetext{These authors contributed equally to this work}\def\thefootnote{\arabic{footnote}}

\begin{abstract}
    The polynomial hierarchy has been widely studied in classical complexity theory. In this paper, we will generalize some commonly known results about the polynomial hierarchy to a version of the hierarchy extended to promise problems. This paper proposes new definitions of existential and universal operators for classes of promise problems. Applying these to $\cc{BQP}$, we recover the hierarchy proposed by Gharibian et al. (MFCS 2018). Moreover, using our definition, we give an easy proof of the \emph{collapse} of this hierarchy under a Karp-Lipton-like scenario, which was an open question for the original definition of Gharibian et al.
    

\end{abstract}

\section{Introduction}
Languages are decision problems that have a definite $\cc{YES}$ or $\cc{NO}$ answer for each input string. On the other hand, promise problems are problems that have a definite $\cc{YES}$ or $\cc{NO}$ solution only for a subset of all possible input strings: these inputs are said to satisfy the promise. An algorithm solving a promise problem only returns the correct answer for inputs satisfying the promise and can have arbitrary answers, or not even halt, for the rest of the inputs. When generalizing the results from languages to promise problems, we have to keep this subtlety in mind.

There are many important complexity classes, arising from probabilistic or quantum models of computation, that consist of promise problems, such as Promise-$\cc{BPP}$, Promise-$\cc{BQP}$, Promise-$\cc{QCMA}$, and Promise-$\cc{QMA}$. In fact, the promise versions of these classes are so common that often the notation $\cc{BQP}$ is implicitly used to mean Promise-$\cc{BQP}$ (and likewise for $\cc{QMA}$, etc.). However, one has to be careful with them as some intuitive `\emph{structural}' relations fail to hold~\cite{Goldreich06}. 
In particular, these structural issues have been an obstacle to finding a natural quantum analogue of the polynomial hierarchy from classical complexity theory. A leading candidate for such an analogue was proposed in~\cite{Gharibian18}. However, the authors of this work were unable to show that their quantum polynomial hierarchy has the property of `collapsing' whenever two distinct levels are equal, because of difficulties arising from the use of promise problems. Classically, this is a key feature of the polynomial hierarchy and is one of the elements of the Karp-Lipton theorem. 

To understand the issues with showing the collapsing property for a quantum polynomial hierarchy, let us recall the definition of the classical polynomial hierarchy. The classical polynomial hierarchy is defined as an \emph{existential} and \emph{universal} extension of $\P$, the class of languages recognized by a polynomial time Turing Machine. Each class in the hierarchy is obtained by repeatedly applying the existential and universal \emph{complexity class operators} $\exists \cdot$ and $\forall \cdot$ to the base class of $\P$, with higher levels in the hierarchy being obtained through the application of more quantifiers. The collapsing phenomenon occurs because two consecutive quantifiers of the same type can be merged into one: formally, $\exists \cdot \exists \cdot  = \exists \cdot $ and likewise for $\forall \cdot \forall \cdot = \forall \cdot$. 

The authors of~\cite{Gharibian18} propose a definition of the quantum polynomial hierarchy that is closely modeled on the classical definition. However, they do not define in terms of quantifier operators on a base complexity class: rather, they work directly with existential and universal quantifiers on input strings to a quantum circuit, to define a hierarchy of promise problems. As a result, the classical proof of collapse based on combining consecutive operators does not generalize. Moreover, the authors note, in Remark 17, that the standard definitions of $\exists \cdot$ and $\forall \cdot $ do \emph{not} behave well for classes of \emph{languages} arising from probabilistic or quantum computation, such as $\cc{BPP}$ or $\cc{BQP}$. 

Inspired by~\cite{Gharibian18}, we propose new existential and universal complexity class operators acting on promise-problem classes, such that consecutive operators of the same type can be merged, and such that some other basic structural relations valid for languages still hold. This allows us to extend some results from classical complexity theory to a hierarchy constructed from classes comprising promise problems. In particular, we were able to complete the Karp-Lipton-like theorem for the quantum-classical case proven in~\cite{Gharibian18}.

\paragraph{Related work.}
After the completion of this work, we became aware of two related works, also proving collapse-type results for generalizations of the quantum polynomial hierarchy. The first is~\cite{lockhart2017quantum_iso}, which proposed essentially the same definitions as we do for $\exists \cdot$ and $\forall \cdot$, and proved collapse for the resulting hierarchy with $\BQP$ as the base class, using similar techniques to us. In comparison, our presentation is somewhat more general, as it applies to any hierarchy built using our operators on any reasonable base class of promise problems (including, for instance $\mathsf{BPP}$). The second is~\cite{AGKR23}, which directly proves a Karp-Lipton theorem for the quantum polynomial hierarchy as well as other results about other hierarchies; we state their result more formally in~\Cref{thm:agkr23}. 

\paragraph{Organization.} This paper is structured as follows. In section \secref{sec: defn}, we define the notation and quantifiers relevant to the present paper. Then, in section \secref{sec: properties}, we prove some basic properties relating to the quantifiers. In section \secref{sec: collapse proof}, we use these properties to prove that a hierarchy for promise problems collapses if the existential and universal classes on a fixed level are equal. In section \secref{sec: QPH application}, we further prove the collapse of the quantum polynomial hierarchy of~\cite{Gharibian18} and use it to prove a crucial step in the \emph{quantum-classical variation} of the Karp-Lipton theorem proven in~\cite{Gharibian18}. We finally conclude with possible further applications of this result in section \secref{sec: conclusion}.

\section{Definitions}\label{sec: defn}

\begin{definition}[Promise problem]
A promise problem $A$ is a pair of disjoint subsets of $\{0,1\}^*$. That is, $A = (A_{yes}, A_{no})$ such that $A_{yes} \cap A_{no} = \phi$. The set $A_{yes} \cup A_{no}$ is called the \emph{promise}. A promise class is defined as a set of promise problems.
\end{definition}

Promise problems are a generalization of decision problems. An algorithm solving a promise problem can distinguish between  $\cc{YES}$ and $\cc{NO}$ instances if the promise holds, but there is no guarantee otherwise. This definition can be found in~\cite{Goldreich06}. We can also define the complement of the promise problem.

\begin{definition}[Complement]
The complement of a promise problem $A = (A_{yes}, A_{no})$ is denoted by $\overline{A}$, where $\overline{A} = (A_{no}, A_{yes})$. It is the promise problem with its $\cc{YES}$ and $\cc{NO}$ instances interchanged. Similarly, we can define the complement of the whole promise class $\class{C}$ as $\cc{co}({\class{C}})=\{\overline{A}|A\in \class{C}\}$.
\end{definition}
Complements follow these properties:
\begin{enumerate}
    \item $\cc{co}(\cc{co}(\class{C})) = \class{C}$.
    \item $\cc{co}({\class{C}}) = \cc{co}({\class{C'}}) \iff \class{C} = \class{C'}$.
\end{enumerate}

\begin{definition}[Existential quantifier] Let $\class{C}$ be a class of promise problems. $\exists \cdot \class{C}$ is defined as the class of promise problems such that a promise problem $L = (L_\text{yes}, L_\text{no}) \in \exists \cdot \class{C}$, if and only if there is polynomial $p$ and a promise problem $A= (A_\text{yes}, A_\text{no}) \in \class{C}$ such that for input $x$
\begin{eqnarray}
x \in L_\text{yes} \iff \exists y \text{ s.t. } (|y|= p(|x|)) \text{ and } \abrak{x,y}\in A_\text{yes}\nonumber\\
x \in L_\text{no} \iff \forall y \text{ s.t. } (|y|= p(|x|)) \text{, } \abrak{x,y}\in A_\text{no}.
\end{eqnarray}

\end{definition}

This definition is similar to those in~\cite{Allender90countinghierarchies, Wrathall1976CompleteSA}, albeit generalized to promise problems. Moreover, we set the length of the quantified string $y$ to be exactly a polynomial of input length. This is equivalent to allowing $y$ to be smaller as one can pad it to reach the required length (as long as $\class{C}$ is a reasonable model with at least logarithmic space). 

Universal quantifiers can be defined by replacing $\exists \leftrightarrow \forall$ in all the above locations.

\begin{definition}[Universal quantifier] Let $\class{C}$ be a class of promise problems. $\forall \cdot \class{C}$ is defined as the class of promise problems such that a promise problem $L = (L_\text{yes}, L_\text{no}) \in \forall \cdot \class{C}$, if and only if there is polynomial $p$ and a promise problem $A= (A_\text{yes}, A_\text{no}) \in \class{C}$ such that for input $x$
\begin{eqnarray}
x \in L_\text{yes} \iff \forall y \text{ s.t. } (|y| = p(|x|)) \text{, } \abrak{x,y}\in A_\text{yes}\nonumber\\
x \in L_\text{no} \iff \exists y \text{ s.t. } (|y| = p(|x|)) \text{ and } \abrak{x,y}\in A_\text{no}.
\end{eqnarray}
\end{definition}

\begin{definition}[Hierarchy of Promise Classes]
For a base promise class $\mathcal{C}$, let us define the first existential level as $\Sigma_1 = \exists \cdot \mathcal{C}$. Correspondingly, the first universal level is defined as $\Pi_1 = \forall \cdot \mathcal{C}$. Higher levels of the hierarchy can be defined inductively as (for $i>1$)
\begin{equation}
    \Sigma_i = \exists \cdot \Pi_{i-1};\qquad \qquad  \Pi_i = \forall \cdot \Sigma_{i-1}
\end{equation}
The hierarchy of $\mathcal{C}$ denoted by $\mathcal{H}$ is defined as the union over all the existential levels.
\begin{equation}
    \mathcal{H} = \bigcup_{i \in \mathbb{N}}\Sigma_{i}
\end{equation}
\end{definition}

Below we will prove the conditional collapse of quantum polynomial hierarchy. More specifically, we will prove the below lemma.
\begin{lemma}\label{lem: eq implies collapse}
    Let $\Sigma_i $ and $\Pi_i$ be $i^\text{th}$ existential and universal level of a hierarchy $\mathcal{H} \equiv \bigcup_{m \in \mathbb{N}}\Sigma_{m}$ of some promise class $\mathcal{C}$. If $\Sigma_i = \Pi_i $, then $\mathcal{H} = \Sigma_i$.
\end{lemma}

Before this, let's look at some properties of the quantifiers that we introduced.

\section{Properties of Quantifiers}\label{sec: properties}
The above-defined quantifiers follow some basic structural relations. We list them below and prove them in this section.
\begin{enumerate}
  \item $\cc{co}(\exists \cdot \mathcal{C}) = \forall\cdot \cc{co}(\mathcal{C}).$
    \item $\exists \cdot \exists \cdot \mathcal{C} = \exists \cdot \mathcal{C}.$
    \item $\forall \cdot \forall \cdot \mathcal{C} = \forall \cdot \mathcal{C}.$
    \item $\mathcal{C}_1=\mathcal{C}_2\implies \exists\cdot \mathcal{C}_1 = \exists\cdot \mathcal{C}_2.$
    \item $\mathcal{C}_1=\mathcal{C}_2\implies \forall\cdot \mathcal{C}_1 = \forall \cdot \mathcal{C}_2.$
\end{enumerate}

Let us prove each of the properties.
\begin{lemma}\label{lem: complement}
$\cc{co}(\exists \cdot \mathcal{C}) = \forall\cdot \cc{co}(\mathcal{C}).$
\end{lemma}
\begin{proof}
Consider a promise problem $P \in \cc{co}(\exists \cdot \mathcal{C})$. This means $\bar{P} \in \exists \cdot \mathcal{C}$. Let the corresponding promise problem in $\mathcal{C}$ be $A$. We have 
\begin{eqnarray}
x \in P_\text{no} \iff \exists y,  \abrak{x,y}\in A_{\text{yes}}\nonumber \\
x \in P_\text{yes} \iff \forall y,  \abrak{x,y}\in A_{\text{no}},
\end{eqnarray}
where the quantifiers are over strings $y$ whose length is equal to $p_{P,A}(|x|)$ for a fixed polynomial $p_{P,A}$ of $x$. 
Let $B = \bar{A}  \in \cc{co}(\mathcal{C})$. We can rewrite the above implications as 
\begin{eqnarray}
x \in P_\text{yes} \iff \forall y,  \abrak{x,y}\in B_{\text{yes}}\nonumber
\\
x \in P_\text{no} \iff \exists y,  \abrak{x,y}\in B_{\text{no}},
\end{eqnarray}
with $y$ having the same length constraint.
This means that $P \in  \forall\cdot \cc{co}(\mathcal{C})$. Hence, 
\begin{equation}
    \cc{co}(\exists \cdot \mathcal{C}) = \forall\cdot \cc{co}(\mathcal{C}).
\end{equation}
\end{proof}

\begin{lemma}\label{lem:exists_exists}
$\exists \cdot \exists \cdot \mathcal{C} = \exists \cdot \mathcal{C}.$
\end{lemma}
\begin{proof}
Consider $P=(P_\text{yes}, P_\text{no}),$ a promise problem such that $P \in \exists \cdot \exists \cdot \mathcal{C}$. We will show that $P 
\in \exists \cdot \mathcal{C}$. We know as $P \in \exists \cdot \exists \cdot \mathcal{C}$, so there is a corresponding promise problem $A \in \exists \cdot \mathcal{C}$ which in turn has a corresponding promise problem $B \in \mathcal{C}$. These promise problems are associated with polynomial functions $p_{P, A}$ and $p_{A, B}$ respectively, bounding the length of the quantified variables. Expanding out the definition of $\exists \cdot \exists \cdot \mathcal{C}$, we obtain the following implications
\begin{eqnarray}
x\in P_\text{yes} \iff \exists y_1 (\abrak{x,y_1}\in A_{\text{yes}})\iff \exists y_1 (\exists y_2 \abrak{x,y_1, y_2}\in B_{\text{yes}})\\
x\in P_\text{no} \iff \forall y_1 (\abrak{x,y_1}\in A_{\text{no}})\iff \forall y_1 (\forall y_2 \abrak{x,y_1, y_2}\in B_{\text{no}}),
\end{eqnarray}
where the quantification is over $y_1$ such that $|y_1| = p_{P,A}(|x|)$ and $y_2$ such that $|y_2| = p_{A,B}(|x| + |y_1|)$.

From the above implications, we can construct a promise problem $P' \in \exists \cdot \mathcal{C}$ and take $y = y_1 || y_2$. In other words, we can concatenate both strings. The concatenated string has length
\begin{eqnarray}
|y_1| + |y_2| &=& p_{P,A}(|x|) + p_{A,B}(|x| + |y_1|) \\
&=& p_{P,A}(|x|) + p_{A,B}(|x| + p_{P,A}(|x|)) \\
&\defeq& p_{P',B}(|x|)
\end{eqnarray}
where $p_{P',B}(\cdot)$ is still a polynomial function. We know that with this concatenation,
\begin{equation}
    \exists y_1 \exists y_2 \simeq \exists y \qquad \forall y_1 \forall y_2 \simeq \forall y.
\end{equation}
Here the length of $y$ should be exactly equal to $p_{P', B}(|x|)$.
Having $B$ as the base promise problem for $P'$, we get
\begin{eqnarray}
x\in P'_\text{yes}\iff \exists y,  \abrak{x,y}\in B_{\text{yes}} \iff\exists y_1 \exists y_2 \abrak{x,y_1, y_2}\in B_{\text{yes}} \iff x\in P_\text{yes}\\
x\in P'_\text{no}\iff \forall y , \abrak{x,y}\in B_{\text{no}}\iff \forall y_1 \forall y_2 \abrak{x,y_1, y_2}\in B_{\text{no}} \iff x\in P_\text{no}.
\end{eqnarray}
Hence, $P = P' \in \exists \cdot \mathcal{C}$. This means that $\exists \cdot\exists \cdot \mathcal{C}\subseteq \exists \cdot \mathcal{C}$.\\
On the other hand, we already know that $\exists \cdot \mathcal{C}\subseteq\exists \cdot \exists \cdot \mathcal{C}$ because we can simply ignore one of the quantified strings. So, 
\begin{equation}
    \exists \cdot\exists \cdot \mathcal{C} =\exists \cdot \mathcal{C}.
\end{equation}

\end{proof}

\begin{corollary}\label{cor: forall_forall}
$\forall \cdot \forall \cdot \mathcal{C} = \forall \cdot \mathcal{C}.$
\end{corollary}

\begin{proof}
We can use Lemma \ref{lem: complement} and \ref{lem:exists_exists} to prove the above corollary. Applying Lemma \ref{lem:exists_exists} on $\cc{co}(\mathcal{C})$, we have
\begin{equation}
    \exists \cdot\exists \cdot\cc{co}(\mathcal{C}) = \exists \cdot \cc{co}(\mathcal{C})
\end{equation}
Take the complement on both sides
\begin{eqnarray}
    \cc{co}(\exists \cdot\exists \cdot\cc{co}(\mathcal{C})) &=& \cc{co}(\exists \cdot \cc{co}(\mathcal{C}))\\
    \forall \cdot \forall \cdot \cc{co}(\cc{co}(\mathcal{C})) &=& \forall \cdot \cc{co}(\cc{co}(\mathcal{C}))
    \\
    \forall \cdot \forall \cdot \mathcal{C} &=& \forall \cdot \mathcal{C}.  
\end{eqnarray}
\end{proof}

\begin{lemma}\label{lem: exists_equality}
$\mathcal{C}_1=\mathcal{C}_2\implies \exists\cdot \mathcal{C}_1 = \exists\cdot \mathcal{C}_2.$
\end{lemma}
\begin{proof}
Consider a promise problem $P \in \exists\cdot \mathcal{C}_1$. We have a corresponding promise problem $A \in \mathcal{C}_1$ such that 

\begin{eqnarray}
x \in P_\text{yes} \iff \exists y \text{ s.t. }\abrak{x,y}\in A_\text{yes}\nonumber\\
x \in P_\text{no} \iff \forall y \text{ s.t. } \abrak{x,y}\in A_\text{no}.
\end{eqnarray}
As $\mathcal{C}_1=\mathcal{C}_2$, it means $A \in \mathcal{C}_2$. So, as the above holds for $A$, it means that $P \in \exists\cdot \mathcal{C}_2$. So, we have shown $\exists\cdot \mathcal{C}_1 \subseteq \exists\cdot \mathcal{C}_2$. By symmetry,
\begin{equation}
    \exists\cdot \mathcal{C}_1 = \exists\cdot \mathcal{C}_2.
\end{equation}
\end{proof}

\begin{corollary}\label{cor: forall_equality}
$\mathcal{C}_1=\mathcal{C}_2\implies \forall\cdot \mathcal{C}_1 = \forall\cdot \mathcal{C}_2.$
\end{corollary}
\begin{proof}
We can use Lemma \ref{lem: complement} and \ref{lem: exists_equality} to prove the above corollary. Applying Lemma \ref{lem: exists_equality} on $\cc{co}(\mathcal{C}_1)$ and $\cc{co}(\mathcal{C}_2)$, we get
\begin{equation}\label{eqn: lem_equality_on_co}
    \cc{co}(\mathcal{C}_1)=\cc{co}(\mathcal{C}_2)\implies \exists\cdot \cc{co}(\mathcal{C}_1) = \exists\cdot \cc{co}(\mathcal{C}_2).
\end{equation}
Firstly, using the properties of complements, we have
\begin{eqnarray}
    \exists\cdot \cc{co}(\mathcal{C}_1) &=& \exists\cdot \cc{co}(\mathcal{C}_2)\\
    \cc{co}(\exists\cdot \cc{co}(\mathcal{C}_1)) &=& \cc{co}(\exists\cdot \cc{co}(\mathcal{C}_2))\\
    \forall \cdot \cc{co}(\cc{co}(\mathcal{C}_1)) &=& \forall \cdot \cc{co}(\cc{co}(\mathcal{C}_2))\\
\forall\cdot \mathcal{C}_1 &=& \forall\cdot \mathcal{C}_2
\end{eqnarray}
So, the equation \ref{eqn: lem_equality_on_co} can be rewritten as 
\begin{equation}
    \mathcal{C}_1=\mathcal{C}_2\implies \forall\cdot \mathcal{C}_1 = \forall\cdot \mathcal{C}_2.
\end{equation}
\end{proof}

\section{Conditional Collapse of Promise Hierarchy}\label{sec: collapse proof}
We may now use the properties proven in section \ref{sec: properties} to prove that a hierarchy formed by any promise class collapses if $\Sigma_i$ and $\Pi_i$ are equal at any level $i$. The hierarchy collapses to the level $i$ at which the two classes are equal. The argument is very similar to the ones used for classical polynomial hierarchy.\\

\begin{proof}[Proof of Lemma 2.1]
    Using Lemma \ref{lem: exists_equality}, we get
\begin{eqnarray}
\exists \cdot \Sigma_{i} &=& \exists \cdot \Pi_{i}\\
\exists \cdot \underbrace{\exists\cdot\forall\cdots}_{i\text{ times}} \mathcal{C}&=&
        \exists \cdot  \underbrace{\forall\cdot\exists\cdots}_{i\text{ times}} \mathcal{C}
\end{eqnarray}
Then, we use Lemma \ref{lem:exists_exists} on the LHS to get 
\begin{equation}
\Sigma_{i} = \Sigma_{i+1}.
\end{equation}
Similarly, using Corollary \ref{cor: forall_forall} and Corollary \ref{cor: forall_equality}, we get
\begin{equation}
    \Pi_{i+1} = \Pi_{i},
\end{equation}
which from initial assumption gives us
\begin{equation}
    \Pi_{i+1} = \Pi_{i} = \Sigma_{i} = \Sigma_{i+1}.
\end{equation}
So, from $\Sigma_i = \Pi_i$, we have shown that $\Sigma_{i+1} = \Sigma_{i}$ and $\Pi_{i+1} = \Sigma_{i+1}$.
Now, we can use induction to show that for any $m \geq i$,
\begin{equation}
    \Sigma_{m} = \Sigma_{i}.
\end{equation}
Hence, taking union over all $m$, we get
\begin{equation}
    \mathcal{H} \equiv \bigcup_{m \in \mathbb{N}}\Sigma_{m} = \Sigma_{i}.
\end{equation}
So, we show that $\mathcal{H}$ collapses to $\Sigma_i$, if $\Sigma_i = \Pi_i $.
\end{proof}

\section{Applications to Quantum Polynomial Hierarchy}\label{sec: QPH application}

Building on the earlier general collapse proof for any promise class, we can adapt this to Promise-$\BQP$. For the rest of this paper, we use the notation from~\cite{Gharibian18}. 
The authors of the paper denote the $i^\text{th}$ existential and universal level of the quantum polynomial hierarchy by $\Qs{i}$ and $\Qf{i}$. Similarly, $\QCPH$ is the corresponding version of the Quantum polynomial hierarchy. It can be seen by inspecting their definition that their $\Qs{i}$ and $\Qf{i}$ are exactly our $\Sigma_i$ and $\Pi_i$ applied to Promise-$\BQP$. To illustrate this, we show formal equality below for $\Qs{i}$ and our $\Sigma_i$. The proof for $\Qf{i}$ and $\Pi_i$ is analogous.
\begin{definition}
    For $0 \leq s < c \leq 1$, let Promise-$\BQP(c,s)$ be the class of promise problems $(A_{yes}, A_{no})$ for which there exists a poly-time quantum algorithm $V$ such that for all $x \in A_{yes}$, the probability that $V$ accepts $x$ is at least $c$, and for all $x \in A_{no}$, the probability that $V$ accepts $x$ is at most $s$.
\end{definition}

\begin{lemma}
The class $\Qs{i}(c,s)$ of~\cite{Gharibian18} is equal to $\Sigma_i$ of the hierarchy built from Promise-$\BQP(c,s)$.
\end{lemma}
\begin{proof}
    By definition, $\Qs{i}(c,s)$ is the class of promise problems $(A_{yes}, A_{no})$ such that there exists a polynomial function $p$ and a poly-time quantum algorithm $V$ for which
    \begin{eqnarray}
    x \in A_{yes} \iff \exists y_1 \forall y_2 \dots Q_i y_i, \; \Pr[V(x, y_1, \dots, y_i)] &\geq c \\
    x \in A_{no} \iff \forall y_1 \exists y_2 \dots \overline{Q_i} y_i, \; \Pr[V(x, y_1, \dots, y_i)] &\leq s,
    \end{eqnarray}
    where the length of each $y_j$ for $j = 1, \dots, i$ is at most $p(|x|)$.

    By definition $\Sigma_i = \exists \cdot \forall \cdot \dots \text{Promise-}\BQP(c,s)$. Unpacking the definition, we see that a promise problem $(A_{yes}, A_{no}) \in \Sigma_i$ iff there exists a promise problem $(B_{yes}, B_{no}) \in \text{Promise-}\BQP(c,s)$ and a collection of polynomial functions $p_1, \dots, p_i$ such that 
     \begin{eqnarray}
    x \in A_{yes} \iff \exists y_1 \forall y_2 \dots Q_i y_i, \; \abrak{x, y_1, \dots, y_i} &\in  B_{yes}, \\
    x \in A_{no} \iff \forall y_1 \exists y_2 \dots \overline{Q_i} y_i, \; \abrak{x, y_1, \dots, y_i} &\in B_{no},
    \end{eqnarray}
    where for each $j = 1, \dots, i$, $|y_j| = p_j(|x| + \sum_{k=1}^{j-1} |y_{k}|)$.

    We will first show the containment $\Qs{i}(c,s) \subseteq \Sigma_i$. Let $(A_{yes}, A_{no})$ be a promise problem in $\Qs{i}(c,s)$ with $V$ the associated quantum poly-time algorithm. For each $j$, let $p_j = p$. Let $(B_{yes}, B_{no})$ be the promise problem defined by
    \begin{eqnarray} 
    \abrak{x, y'_1, \dots, y'_i} \in B_{yes} &\iff \Pr[V(x,y_1, \dots, y_i) \geq c] \\
    \abrak{x,y'_1, \dots, y'_i} \in B_{no} &\iff \Pr[V(x,y_1, \dots, y_i) \leq s],
    \end{eqnarray}
    where $y'_j$ is $y_j$ padded to length $p(|x| + \sum_{k=1}^{j-1} |y'_k|)$. 
    
    By construction, $(B_{yes}, B_{no}) \in \text{Promise-}\BQP(c,s)$. Now, taking $(B_{yes}, B_{no})$ as the `base promise problem' and applying the definition of $\Sigma_i$, we see that $(A_{yes}, A_{no}) \in \Sigma_i$.

    Now, let us show the containment the other way, i.e. the containment $\Sigma_i \subseteq \Qs{i}$. The argument is very similar. Let $V'$ be the Promise-$\BQP(c,s)$ algorithm for $(B_{yes}, B_{no})$, and let $p$ be the maximum length of any $y_j$; this is guaranteed to be polynomial in $|x|$. Then set $V$ to be the algorithm $V'$ applied to inputs $\abrak{x,y_1,\dots, y_i}$ where each $y_j$ has been padded to length $p(|x|)$. Applying the definition of $\Qs{i}(c,s)$ to $V$ and $p$, we see that $(A_{yes}, A_{no}) \in \Qs{i}(c,s)$.
\end{proof}

In particular, with Lemma \ref{lem: eq implies collapse}, setting the base class as Promise-$\BQP$, we immediately get the collapse of $\QCPH$ whenever any two levels $\Qs{i} = \Qf{i}$ are equal.

\begin{corollary} \label{cor:qcma-bqp}
    If $\Qs{i} = \Qf{i}$, then $\QCPH = \Qs{i} = \Qf{i}$.
\end{corollary}
In particular, we have the following consequence of an equality between levels $1$ and $0$ of the hierarchy.
\begin{corollary}
If $\QCMA = \BQP$, then $\QPH = \BQP$.
\end{corollary}
\begin{proof}
    If $\QCMA = \Qs{1} = \BQP$, then by applying $\cc{co}$ to both sides, we get $\Qf{1} = \cc{co}\BQP = \BQP$, and hence $\Qs{1} = \Qf{1}$. Applying the previous corollary finishes the proof.
\end{proof}
The authors of~\cite{Gharibian18} illustrate that if $\pQCMA \subseteq \BQP_{/\text{mpoly}}$, then $\Qs{2} = \Qf{2}$. However, they note the challenge of using this equality for collapsing the $\QCPH$ due to the non-structural behavior of promise classes. Our paper addresses this subtlety. Specifically, our last two corollaries taken together resolve Conjecture 9 of their paper.

Independently of our work, in a forthcoming paper~\cite{AGKR23}, the authors present their own proof of a Karp-Lipton analogue for $\QCPH$, in addition to results about other versions of the quantum polynomial hierarchy with quantification over quantum states.
\begin{theorem}[\cite{AGKR23}]
\label{thm:agkr23}
If $\QCMA \subseteq \BQP_{/\text{mpoly}}$, then $\QCPH = \Qs{2} = \Qf{2}$.
\end{theorem}

Our approach, notably, is not just applicable to the quantum polynomial hierarchy. The notation and quantifiers defined here can be useful in proving similar results in the probabilistic world of classes or any semantic class based on a non-trivial promise.


\section{Conclusion}\label{sec: conclusion}
In this paper, we proposed a new definition of quantifiers for classes of promise problems. These help alleviate the `\emph{structural}' failure of promise problems. Under this definition, we were able to extend some useful properties of language classes to promise classes and prove the conditional collapse of the hierarchy for promise problems. We further showed some applications to the quantum polynomial hierarchy and worked on the extension of the Karp-Lipton theorem to the quantum-classical case. 

\section{Acknowledgment}
This project was originally a class project for 6.845 (Quantum Complexity Theory) in the Spring of 2022 at MIT.

\bibliographystyle{myhalpha}
\bibliography{reference}

\appendix
\section{Appendix}
In this section, we prove some auxiliary remarks for the existential and universal levels that we defined in this paper. Here, we prove that, for a hierarchy whose base class is closed under complement, if the existential class at one level is a subset of the universal class at the same level, then they are equal.
\begin{lemma}\label{lem: subset implies eq}
If the base class $\mathcal{C}$ is closed under complement then $\Sigma_i \subseteq \Pi_i \implies \Sigma_i = \Pi_i$.
\end{lemma}
\begin{proof}
Assume $\Sigma_i \subseteq \Pi_i$. Here, we can take advantage of the fact that  
\begin{equation}
\mathcal{C}_1 \subseteq \mathcal{C}_2 \iff \cc{co}(\mathcal{C}_1) \subseteq \cc{co}(\mathcal{C}_2)
\end{equation}
We also know that $\mathcal{C}$ is closed under complement. This means $\mathcal{C} = \cc{co}(\mathcal{C})$.
Then, the following equations hold
\begin{eqnarray*}
\Sigma_i &\subseteq& \Pi_i\\ 
\underbrace{\exists\cdot\forall\cdots}_{i\text{ times}} \mathcal{C} &\subseteq& \underbrace{\forall\cdot\exists\cdots}_{i\text{ times}} \mathcal{C}\\
\cc{co}(\underbrace{\exists\cdot\forall\cdots}_{i\text{ times}} \mathcal{C}) &\subseteq& \cc{co}(\underbrace{\forall\cdot\exists\cdots}_{i\text{ times}} \mathcal{C})\\
\underbrace{\forall\cdot\exists\cdots}_{i\text{ times}}\cc{co}(\mathcal{C}) &\subseteq& \underbrace{\exists\cdot\forall\cdots}_{i\text{ times}}\cc{co}( \mathcal{C})\qquad \qquad (\text{From Lemma \ref{lem: complement}})\\
\underbrace{\forall\cdot\exists\cdots}_{i\text{ times}}\mathcal{C} &\subseteq& \underbrace{\exists\cdot\forall\cdots}_{i\text{ times}} \mathcal{C}\\
\Pi_i &\subseteq& \Sigma_i
\end{eqnarray*}
Now, as $\Sigma_i \subseteq \Pi_i\implies \Pi_i \subseteq \Sigma_i$, if $\Sigma_i \subseteq \Pi_i$, 
\begin{equation}
    \Sigma_i = \Pi_i.
\end{equation}
\end{proof}
\end{document}